\documentclass[11pt]{article}

\usepackage{amssymb,amsmath,amsthm,amsfonts,sectsty,url,graphicx,fullpage}
\usepackage[bold,full]{complexity}

\newsavebox{\theorembox}
\newsavebox{\lemmabox}
\newsavebox{\corollarybox}
\newsavebox{\propositionbox}
\newsavebox{\examplebox}
\newsavebox{\conjecturebox}
\newsavebox{\algbox}
\newsavebox{\qbox}
\newsavebox{\problembox}
\newsavebox{\definitionbox}
\newsavebox{\assumptionbox}
\newsavebox{\hypothesisbox}
\savebox{\theorembox}{\noindent\bf Theorem}
\savebox{\lemmabox}{\noindent\bf Lemma}
\savebox{\corollarybox}{\noindent\bf Corollary}
\savebox{\propositionbox}{\noindent\bf Proposition}
\savebox{\examplebox}{\noindent\bf Example}
\savebox{\conjecturebox}{\noindent\bf Conjecture}
\savebox{\algbox}{\noindent\bf Algorithm}
\savebox{\qbox}{\noindent\bf Question}
\savebox{\definitionbox}{\noindent\bf Definition}
\savebox{\problembox}{\noindent\bf Problem}
\savebox{\assumptionbox}{\noindent\bf Assumption}
\savebox{\hypothesisbox}{\noindent\bf Hypothesis}
\newtheorem{theorem}{\usebox{\theorembox}}

\newtheorem{definition}{\usebox{\definitionbox}}
\renewcommand{\Pr}{ \mathrm P}
\newcommand{\IN}{\mbox{$I\!\!N$}}

\newcommand{\citeyear}{\cite}

\begin{document}

\title{On the Non-Existence of Nash Equilibrium in Games with Resource-Bounded Players}

\author{Joseph Y. Halpern
\and Rafael Pass
\and Daniel Reichman}

\maketitle

\begin{abstract}

We consider sequences of games $\mathcal{G}=\{G_1,G_2,\ldots\}$ where,
for all $n$, $G_n$ has the same set of players. Such sequences arise
in the analysis of running time of players in games, in electronic
money systems such as Bitcoin and in cryptographic protocols. Assuming
that one-way functions exist, we prove that there is a sequence of
2-player zero-sum Bayesian games $\mathcal{G}$ such that, for all $n$,
the size of every action in $G_n$ is polynomial in $n$, the utility
function is polynomial computable in $n$, and yet there is no
polynomial-time Nash equilibrium, where we use a notion of Nash
equilibrium that is tailored to sequences of games. We also demonstrate that Nash equilibrium may not exist when
considering players that are constrained to perform at most $T$ computational steps in each of the games $\{G_i\}_{i=1}^{\infty}$. These examples may shed
light on competitive settings where the availability of more running
time or faster algorithms lead to a ``computational arms race",
precluding the existence of equilibrium. They also point to
inherent limitations of concepts such ``best response" and Nash
equilibrium in games with resource-bounded players.

\end{abstract}
\newpage
\section{Introduction}
Computation dealing with large amounts of data is fundamental to many
fields such as data-mining, artificial intelligence and algorithmic
trading. In the ``big-data" era, the need for faster algorithms and
more powerful computing machines is ever-growing, and is amplified by
competition between firms. Competitive settings, where competing
parties keep searching for more efficient algorithms and building more
powerful hardware, lead to a natural question: what is the likely
outcome of such competition in the long run and, in particular, is
some kind of equilibrium likely to be reached?

We model such settings using tools from game theory and complexity
theory. Formally, consider \emph{computational games} \cite{Kuhn}
$\{G_1,G_2,\ldots\}$, where for all $n$, $G_n$ is a finite game. We
assume that each player chooses a Turing machine (TM) that, given $n$,
computes a strategy for the player in $G_n$. Computational games arise
in several settings. One example is ``crypto-currencies" such as
Bitcoin. An essential ingredient of Bitcoin (e.g., \cite{Bitcoin}) is
miners who solve challenging cryptographic problems, whose solution is
later used in verifying transactions in the system. Miners get
electronic cash rewards for solving such puzzles. Bitcoin keeps the
average time at which puzzles are solved a constant by making the
cryptographic problem needed to be solved harder and harder, forcing
miners to examine a larger number of possible solutions. In this way,
it deals with technological advances that allow faster solution to
cryptographic puzzles. Such a scenario can be modeled as a sequence of
games, where in the $n$th game the miner is required to solve a
cryptographic puzzle $P_n$ where the number of candidate solutions
that need to be examined in order to solve $P_n$ is a function of
$n$.

Cryptographic protocols such as \emph{commitment schemes}
\cite{Commitment} provide another example of computational games. A
commitment scheme consists of two parties; a sender and a receiver. In
the first step of this protocol, the sender chooses a bit $b$ and
sends an encryption of $b$ to the receiver, committing the sender to
$b$ without revealing $b$ to the receiver.
Next, the receiver chooses a bit; Finally in the third step, the
sender reveals the bit to the receiver. This protocol can be viewed as
a game where the receiver wins if the
bit he chooses matches the bit revealed; the sender wins if they do
not match. Clearly, if the receiver can break the scheme and
deduce the sender's bit, the receiver wins; if the sender can cheat
(``reveal" a bit that does not necessarily match what he committed
to), the sender wins. The encryption at the first step involves a security
parameter $k$, where larger security parameters provide more security
(i.e.., more running time is required to break the scheme). This can be
modeled as a sequence of games, where in the $k$th game the sender
encrypts the bit using a security parameter $k$
(see also \cite{Kuhn}).
Many cryptographic protocols, including secret sharing and multiparty
computation, can be viewed as games  in this way.%
 \footnote{The games used to model such protocols are
  extensive-form games. Our results demonstrate non-existence of
  equilibrium in Bayesian games, which are a special case of
  extensive-form games. Hence, our non-existence results carry over to
  extensive-form games.}

We focus on \emph{resource-bounded} players, continuing a long line of
work in game theory (e.g., \cite{Megiddo,Neyman,Neyman2,Yan}). We
begin by studying players that are polynomial-time bounded. For a
sequence $\mathcal{G}=\{G_1,G_2,\ldots\}$, this requirement entails
that for all $n$, both the actions performed by players in $G_n$ and the
computation of utilities can be done in time polynomial in $n$. We
also assume that the length of every action is polynomial in $n$. We
call sequences of games with these properties \emph{polynomial games}.
These games seem to capture intended applications such as Bitcoin
and cryptographic protocols.

One of the most widely used solution concepts in game theory is Nash
equilibrium (NE). In a Nash equilibrium, no player can improve his
utility by deviating unilaterally from his strategy. NE can be
extended to polynomial games $\{G_1,G_2,\ldots\}$ \cite{Kuhn}. We
assume that every game $G_j$ is a $k$-player game and that for $1\le i
\le k$, player $i$ uses a TM $M_i$ that computes his actions in
$G_j$ given $j$. Roughly speaking, a machine profile $(M_1,\ldots,M_k)$
consisting of polynomial time TMs is a polynomial-time NE
if, for every player, replacing his TM by a different
(polynomial-time) TM gives him at most a negligible improvement to his
utility. There are certain subtleties in this definition; see
Definition \ref{def:machine} and the discussion thereafter for more
detail.  We show that, assuming the existence of one-way functions, there are
polynomial 2-player zero-sum games for which no polynomial-time Nash
equilibrium exists. The idea is to simulate the ``largest integer
game" in this setting, the game where players simultaneously output an
integer, and the player who chooses the largest integer wins. Clearly
this game has no Nash equilibrium \cite{Solan}. We can effectively
simulate this game by presenting players with multiple one-way
function puzzles, requiring players to invert as many puzzles as
possible. We can ensure that a player with sufficiently (polynomially)
more running time can invert more puzzles. Thus, we get an ``arms
race" with no equilibrium. This example points to an inherent
difficulty in analyzing games with polynomially bounded
players. Namely, in such games there is often no best response;
players can use longer and longer running times to improve their
payoffs.
Interestingly, a similar phenomenon has been observed in Bitcoin, where miners use increasingly more sophisticated computational devices for the mining operation (see \cite{Bitcoin} and the reference within).

Finally, we demonstrate that Nash equilibrium may fail to exist even
if players are constrained to run for at most $T$ steps for a fixed
integer $T$, without asymptotics kicking in. The idea is to let
players first play a game (matching pennies) that requires
randomization to achieve equilibrium, and then effectively give the
player with greater remaining running time an additional
bonus. Assuming that the generation of a random bit requires computational
effort, this game cannot have a Nash equilibrium.

We may hope that, in our examples, even if there does not exist a NE,
there might exist an $\epsilon$-NE for some small $\epsilon$. However
all our arguments for the non-existence of NE also show the
non-existence of $\epsilon$-NE for some appropriate $\epsilon>0$.

Polynomial games bear some similarities to \emph{succinct games}. In
succinct games,
there exists a circuit $C$ that calculates the utility $C(x_1,x_2,\ldots,x_k)$ of the
players once they choose the actions $x_1,x_2,\ldots,x_k \in
\{0,1\}^m$. It is known that, given a 2-player zero-sum succinct game,
it is EXP-hard to find a NE \cite{Koller,Fortnow} (see also
\cite{Vadhan}). Our results regarding the non-existence of NE in
polynomial games are incomparable to these results. We are concerned
with \emph{polynomial-time computable} strategies. Considering
polynomially bounded players (as opposed to unbounded players) may
drastically change the set of Nash equilibria in succinct
games. Indeed, a NE for a sequence of games $\{G_1,G_2,\ldots\}$ with
polynomially bounded players may fail to be a Nash equilibrium for
$G_n$ for all $n\geq 1$: for an example, see the end of
Section~\ref{sec:oneway}. For similar reasons, the PPAD-hardness
results of finding a Nash equilibrium in a fixed game \cite{Chen,Pap}
cannot be applied in our setting either.

\section{Preliminaries}
We begin by defining Bayesian games.

\begin{definition}
A $k$\emph{-player normal-form Bayesian game} is described by a tuple
$(J,B,T,P,v)$, where
\begin{itemize}
\item $J$ is a set of $k$ players (we identify $J$ with $[k]=\{1,\ldots,k\}$);
\item $B=\prod_{i=1}^kB_i$, where $B_i$ is a finite set for all $i \in
  [k]$ consisting of the available actions of player $i$;
\item $T=\prod_{i=1}^kT_i$, where $T_i$ is a finite set called the \emph{type space} of player $i$;
\item $P$ is a probability distribution over $T$;
\item $v=(v_1,\ldots,v_k)$, where for all $i$, $v_i$ is a function from
  $B \times T$ to the real numbers.
\end{itemize}
\end{definition}

In our settings, it will often be the case that all types are
perfectly correlated: all players have the same type and all players
know the type of every other player. Observe that normal-form games
can be viewed as a special case of Bayesian games (where the type
space is a singleton).  Finally, since we are concerned here mainly
with Bayesian
games, when we write ``game" we mean ``Bayesian game", unless
explicitly stated otherwise.

A \emph{pure strategy} $s_i$ for player $i$ is a map $s_i:T_i\rightarrow B_i$; a strategy $s_i$ maps the type $t_i \in T_i$ of player $i$ to an action $s_i(t_i) \in B_i$.
We denote by $\Delta(B_i)$ the set of all probability distribution
over $B_i$; let $\Delta=\Pi_{i=1}^k\Delta(B_i).$ A mixed-strategy
$s_i$ for player $i$ is a function mapping type $t_i \in T_i$ to an
element of $\Delta (B_i)$. We denote by $s_i(t_i,b_i)$ the probability
assigned by a mixed strategy $s_i(t_i)$ to $b_i \in
B_i$. The expected utility of player $i$ with the mixed strategy
profile $s=(s_1,\ldots,s_k)$ (where $t =(t_1,\ldots,t_k)\in T$,
$b=(b_1,\ldots,b_k) \in B$, and $(s_1(t_1),\ldots,s_k(t_k)) \in \Delta)$ is
given by
\begin{equation}\label{eq1}
V_i(s)=\sum_{t \in T}P(t)\sum_{b \in
  B}\left(\prod_{i=1}^ks_i(t_i,b_i)\right) v_i(t,b).
\end{equation}
Note that there are two sources of uncertainty in the utility of a
player choosing a mixed action: the probability distribution over
other players actions and the distribution $P$ over the type space.

\begin{definition}\label{definition:Nash}
Let $G=(J,B,T,P,v)$ be a $k$-player Bayesian game and suppose that
$\epsilon \geq 0$. A mixed-strategy profile $s=(s_1,\ldots,s_k)$ is an
$\epsilon$-\emph{Nash equilibrium} ($\epsilon$-NE for short) if,
for all players $i$ and all mixed strategies $s'_i$, we have that
$$V_i(s)\geq V_i(s'_i,s_{-i})-\epsilon.$$
(As usual, if $s=(s_1,\ldots,s_k)$ then
$s_{-i}=(s_1,\ldots,s_{i-1},s_{i+1},\ldots,s_k)$ is the tuple excluding.
$s_i$.) When $\epsilon=0$, we have a Nash equilibrium.
\end{definition}

To reason about resource-bounded players in games, we consider a
sequence $\{G_1,G_2,\ldots\}$ of games where, for all $n$,
$G_n=(J,B^n,T^n,P^n,v^n)$ is a $k$-player game ($k$ is fixed and does
not depend on $n$).
We adapt the definition of \cite{Kuhn}, which in turn is based on
earlier definitions by Dodis, Halevy and Rabin \citeyear{Dodis} and Megiddo and
Wigderson \citeyear{Megiddo}, and is applied to extensive-form games,
to Bayesian games. For an integer $s$,
recall that
$\{0,1\}^{\leq s}$ is the set of all bit strings of length at most $s$.

\begin{definition}\label{defintion:sequence}
A computational game $\mathcal{G}=\{G_1,G_2,\ldots\}$ is a sequence of
normal-form Bayesian games, where $G_n = ([k],B^n,T^n,P^n,v^n)$
, such that
\begin{itemize}
\item The set of players in $G_n$, $[k],$ is the same for all $n$.
\item For all $n$ and all $i$, $B_i^n\subseteq\{0,1\}^{\leq m}$ for some finite
 $m$ (that may depend on $n$).
\item For all $n$ and all $i$, $T_i^n \subseteq\{0,1\}^{\leq r}$ for some finite
 $r$ (that may depend on $n$).
\item For all $i \in [k]$ and $n$, there is a TM $M$ such that, given
  $b \in B^n$, $t \in T^n$, and $1^n$, computes $v^n_i(b,t)$.
\end{itemize}
$\mathcal{G}$ is \emph{bounded} if there exist constants $0<c<C$ such
  that for all $n,b \in B^n$, and $t \in T^n$ we have that $v^n_i(b,t)
  \neq 0\Rightarrow |v^n_i(b,t)| \in [c,C].$
\end{definition}

When dealing with games with polynomial-time players, we require
slightly stronger properties summarized in the definition
below. Following the definition of polynomial games for
extensive- form games \cite{Kuhn}, we define polynomial games for a sequence of Bayesian games.
\begin{definition}\label{defintion:polynomial}
A computational game $\mathcal{G}=\{G_1,G_2,\ldots\}$ is a \emph{polynomial game} if the following conditions hold:
\begin{itemize}
\item There exist a polynomial $p$ such that, for all $n$ and all $i$, $B_i^n=\{0,1\}^{\leq p(n)}$.
\item There exist a polynomial $q$ such that, for all $n$ and all $i$, $T_i^n=\{0,1\}^{\leq q(n)}$.
\item For all $i \in [k]$ and $n$, there is a TM $M$ such that, given $b=(b_1,\ldots,b_k) \in B^n$, $t \in T^n,$ and $1^n$, computes $v^n_i(b,t)$ and runs in time polynomial in $n$.
\end{itemize}
\end{definition}

A \emph{strategy} for player $j$ in a computational game $\mathcal{G}$ is a
TM $M_j$ that, given $1^n$ and the type $t_j \in T_j^n$, outputs a randomized
strategy $M_j(1^n,t_j)$ for player $j$ in game $G_n$, that is,
probability distribution over $B_j^n$.  $M_j(1^n)$ is the
strategy defined by taking $M_j(1^n)(t_j)=M_j(1^n,t_j)$.
The utility of player $i$ in $G_n$ given a machine profile $(M_1 \ldots M_k)$ is
$V_i^n(M_1(1^n),\ldots,M_k(1^n))$ (as defined in (\ref{eq1})).

To analyze computational games $\mathcal{G}=\{G_1,G_2,\ldots\}$, we
would like to be able to apply classical game-theoretic notions, such
as best response and Nash equilibrium, to sequences of games. However,
there are certain difficulties in generalizing these notions to
computational games. A first obstacle is that sequences of infinite
games may allow resource-bounded players to improve over any strategy
by doing additional polynomial-time computations. For example,
consider a player who gets a payoff of 1 by breaking an encrypted
massage $E(s)$ with $s \in \{0,1\}^n$ and a payoff of $0$ if he does
not break it, where the player's running time is polynomial in
$n$. Assuming that there is no polynomial-time algorithm (in $n$) for
finding $s$ given $E(s)$, there is no best response in this game, as a
player can always make polynomially many additional ``guesses" on top
of his current action, increasing his expected utility.
As pointed out by Dodis, Halevi, and Rabin \citeyear{Dodis}, this
observation applies to many problems of interest, such as those
arising from cryptographic protocols \cite{Dodis}.

One way around this problem, suggested in \cite{Dodis,Kuhn}, is to
ignore \emph{negligible} additive changes in
the utility of players, where a sequence $\delta(n)$ is negligible if
for every polynomial $p,p(n)=o(\delta(n)^{-1})$. That is, deviations
that result in a negligible increase in utility are not
considered to be improvements.
Ignoring negligible terms suffices to ensure the existence of
equilibrium in a number of games of interest for which there would not
be an equilibrium otherwise \cite{Dodis}.

If we ignore negligible change, then given a machine profile
$\overline{M}$, changing the behavior of a  TM $M$ in finitely many
games will not be a
deviation breaking an alleged equilibrium,  as altering a sequence
$\delta(n)$ on finitely many $n$'s does not change the fact that
$\delta(n)$ is negligible. On the other hand, a deviation that
improves a given player utility on infinitely many $n$'s by
improves a given player's utility on infinitely many $n$'s by
a constant $\delta>0$ implies that the machine profile is not a
NE. Finally, it is worth noting that if the utilities of players are
exponentially small (say, on the order of $1/2^n$ in the game $G_n$),
a negligible
additive term can have a noticeable effect on the utility of
players; on the other hand, if utilities are exponentially large, even
a (non-negligible) constant change in utilities would be viewed as negligible.
In order to avoid such scaling issues, we deal exclusively
with bounded games when considering solution concepts for
computational games.

\begin{definition}\label{def:machine}
Let $\mathcal{M}$ be a set of TMs and let $\epsilon\geq 0$ be a constant independent of $n$.
A profile $\overline{M}=(M_1,\ldots,M_k)$ of TMs is an $\epsilon$-$\mathcal{M}$-NE for a bounded computational game $\mathcal{G}$ with respect to $\mathcal{M}$, if (a) for all $i$,
$M_i \in \mathcal{M}$, and (b) there exists a negligible sequence $\delta(n)$ such that, for all $M_i' \in \mathcal{M}$ and all $n>0$ and all $i \in [k]$ we have that
\begin{equation}
\label{eq:Nash}
V_i(M_i(1^n),M_{-i}(1^n))\geq V_i({M'_i}(1^n),M_{-i}(1^n))-\epsilon-\delta(n).
\end{equation}
When $\epsilon=0$, we say that $\overline{M}$ is a $\mathcal{M}$-Nash equilibrium. If $\mathcal{M}$ is the set of all probabilistic polynomial-time TMs, we say $\overline{M}$ is a \emph{polynomial} $\epsilon$-NE.
\end{definition}

We can discuss polynomial-time players, best response, and
equilibrium even if the action space of every player is of
super-polynomial size. However, in this case, there are trivial
examples showing that a NE may not exist. For example, one can take
$G_n$ to be the 2-player zero-sum game where each player outputs an
integer of length at most $2^{2^n}$ (written in binary) and the player
outputting the larger integer receives payoff 1, with both players
getting $0$ in case of equality. Clearly this sequence of games does
not have a polynomial equilibrium.

\section{Games With No Polynomial Equilibrium}\label{sec:oneway}

As we now show, there is a polynomial game for which there is no polynomial NE, assuming one-way functions exist.
We find it convenient to use the definition of one-way function given in \cite{Holenstein}.

\begin{definition}
\label{definition_oway}
Given $s: \IN \rightarrow \IN$, $t: \IN \rightarrow \IN$,
a \emph{one-way function with security parameter $s$ against a
$t$-bounded inverter} is a family of functions
$f_k:\{0,1\}^k\rightarrow \{0,1\}^m$, $k = 1, 2, 3, \ldots$,
satisfying the following
properties:
\begin{itemize}
\item
$m=k^b$ for some positive constant $b$;
\item there is a TM $M$ such that, given $x$ with $|x|=k$
  computes $f_k(x)$  in time polynomial in $k$;
\item for all but finitely many $k$'s and all probabilistic TM $M'$,
running in time at most $t(k)$ for a given input $f_k(x)$,
$$\Pr[f_k(M'(f_k(x)))=f_k(x)]<\frac{1}{s(k)},$$
where the probability $\Pr$ is taken over $x$ sampled uniformly from
$\{0,1\}^k$ and the randomness of $M'$.
\end{itemize}
\end{definition}
We assume that exponential one-way functions exist. Specifically, we
assume that there exists a one-way function that is $2^{k/10}$-secure
against a $2^{k/30}$-bounded inverter.
The existence of a one-way function with these parameters follows from
an assumption made by Wee \citeyear{Wee} regarding the existence of
exponential non-uniform one-way functions.
Given $f_k(x)$, we say an algorithm \emph{inverts} $f_k(x)$ if it
finds some $z$ such that $f_k(x)=f_k(z)$.

We can now demonstrate the non-existence of polynomial-time computable
equilibrium in a polynomial game.
\begin{theorem}\label{thm:one_way}
If there exists a one-way function that is $2^{k/10}$-secure against a
$2^{k/30}$-inverter, then, for all $\epsilon > 0$, there exists
a 2-player zero-sum polynomial  game $\mathcal{G}$ that
has no polynomial $\epsilon$-NE.
\end{theorem}
\begin{proof}
We first prove the theorem for $\epsilon=0$ and then show how to extend it
to arbitrary $\epsilon>0$.
Let $\mathcal{G}=\{G_1,G_2,\ldots\}$ be the following polynomial game,
which we call the \emph{one-way function game}. For all $n$, we
define $G_n$ as follows. There are two players, 1 and 2.
Fix a one-way function $\{f_k\}_{k \geq 1}$ that is
$2^{k/10}$-secure against a $2^{k/30}$-bounded inverter. The type
space is the same for each player, and consists of tuples of
$l=\lceil \log n \rceil $ bitstrings of the form $(f_{\lceil \log
  n \rceil}(x_1),\ldots,f_{\lceil \log n \rceil^{2}}(x_l))$.
The
distribution on types is
generated by choosing $x_i \in \{0,1\}^{i\lceil\log n\rceil^{}}$ uniformly at random,
and choosing the $x_i$'s independently.  Given his type
$t_n = (f_{ \lceil \log   n\rceil^{} }(x_1),\ldots,f_{\lceil \log
  n\rceil^{2}}(x_l)\}$, player $j$
outputs $y^j_1,\ldots,y^j_l$. A {\em
  hit} for player $j$ is an index $i$ such that $f_{i \lceil \log
  n\rceil^{}}(y^j_i)=f_{i \lceil \log n \rceil^{}}(x_i)$.
Let $a_j$ denote how many hits player $j$ gets.
The payoff of player $j$ is $1$ if
$a_j-a_{3-j}>0$. If $a_j-a_{3-j}=0$, both players receive a payoff of
0. Observe that the utility function of each player is polynomial-time
computable in $n$. Clearly the length of every action of $G_n$ is
polynomial in $n$ and so is the length of the type $t_n$. Hence the
one-way function game is a polynomial game.

We now prove that there cannot be a polynomial-time NE for $\mathcal{G}$.
game. Suppose, by way of contradiction, that there is a NE $(M_1,M_2)$
for $\mathcal{G}$.  Choose $r > 1$ such that the running time of both
$M_2$ and $M_2$ is bounded by $n^r$ for sufficiently large $n$.

We claim that, for all $n$ sufficiently large, for all $i \ge 60r$,
the probability that $M_j$ inverts $f_{i \lceil \log n\rceil^{}}(x_i)$ is
at most $\frac{1}{2^{i\lceil \log n\rceil/10}}$.
To prove this claim, suppose by way of contradiction that if there exists
$i \ge 60r$ such that,
infinitely many $n$'s, $M_j$ inverts $f_{i \lceil \log
  n\rceil^{}}(x_i)$ with
too large a probability, then we get a contradiction to
$\{f_k\}_{k \geq 1}$ being a one-way function. The idea is to use
$M_j$ to construct a TM $M_j'$ that can invert $f_k(x_k)$ for infinitely many
$k$'s with probability that is larger than $\frac{1}{s(k)}.$

First suppose that given $x$
with $|x|=s \lceil\log m \rceil$ where $30r < s \leq \lceil\log m\rceil$, it is the case that
given $1^m$ and $f_{i \lceil\log m\rceil}(x_i)$ for all $1 \leq i \leq \lceil\log m\rceil$, we
have that $M_j$ inverts $f_{s \lceil\log m\rceil}(x_m)$ in time $m^r$ with probability greater than $2^{-s \lceil\log
  m\rceil/10}=2^{-|x|/10}$.
  Then there is a TM $M_j"$ that inverts $f_{s
  \lceil\log m\rceil}(x)$ in time at most $2m^r$ with probability greater than
$2^{-|x|/10}$. $M_j"$ simply generates $f_{i \lceil\log m\rceil}(x_i)$ for all
$i\neq s$ in the range $[1 \ldots \lceil\log m\rceil]$, where $x_i \in
\{0,1\}^{i\lceil\log m\rceil}$ is selected uniformly at random. By the definition
of one-way function, generating all these $x_i$'s can be done in time
$o(m)$. Next $M_j"$ generates $1^m$ and runs $M_j$ on the sequence of
strings $\{f_{\lceil\log m\rceil}(x_1),\ldots,f_{\lceil\log m\rceil^2}(x_{\lceil\log m}\rceil)\}$ and
$1^m$. The total running of $M_j"$ is at most $m^r+o(m^r)< 2m^r$ and
the probability $M_j"$ inverts $f_{s \lceil\log m\rceil}(x)$ is identical to that
of $M_j$.

Consider now an arbitrary $x \in \{0,1\}^*$. Suppose that $|x|$ is of
the form $s\lceil\log m\rceil$ with $s \leq \lceil\log m\rceil$. Furthermore, assume $s$ is
the largest integer such that $|x|=s \lceil\log m\rceil$ with $s \leq \lceil\log m\rceil$
(equivalently $m$ is the smallest integer such that $|x|=s\lceil\log m\rceil$ for
an integer $s \leq \lceil\log m\rceil$). Clearly there are at most $s \leq \lceil\log m\rceil$
distinct ways to write $|x|$ as a product of integers $s'$ and $t$
such that $t=\lceil\log m'\rceil$ for an integer $m'$ and $s' \leq \lceil\log
m'\rceil$. Suppose now there exist two integers $s',m'$ with $s'\geq 30r+1$
such that $|x|=s'\lceil\log m'\rceil$ and that
$M_j(f_{\lceil\log m'\rceil}(x_1),\ldots,f_{\lceil\log m'\rceil^2}(x_{\lceil\log m'\rceil}),1^{m'})$ runs in time $(m')^r$ and inverts $f_{s \lceil\log m'\rceil}(x)$ with probability larger than $\frac{1}{2^{s'\lceil\log m'\rceil/10}}$. Since $s'>30r+1$, and as $\lceil\log m'\rceil \geq \lceil\log m\rceil$, we have that the running time of $M_j$ given the type and $1^{m'}$ is at most $$(m')^r \leq 2^{s' \lceil\log m'\rceil/30-\lceil\log m\rceil/30}=2^{|x|/30}m^{-1/30}.$$ As before, $M_j$ can be converted to a TM $M_j"$ that simulates $M_j$ when given only $f_{s'\lceil\log m'\rceil}(x_s')$ where the running time of $M_j"$ is at most doubled when compared to the running time of $M_j$.

We now explain how to construct a TM $M'$ that inverts infinitely many inputs $z$ using $M$ and show this implies a violation of our assumption that $f_k$ $k= 1, 2, 3, \ldots$ is a one-way function.
Let $b$ be the constant that is presumed to exist by the definition of
one-way function such that $f_k: \{0,1\}^k\rightarrow \{0,1\}^{k^b}$.
Given an input $z$, $M_j'$ checks if $|z| = N^b$ for some $N$
divisible by $i$.  Clearly this
can be done in time polynomial in $|z|$.  If not, $M_j'$ halts and
outputs some fixed constant, say 0.
(In this case $|z| \ne i \lceil \log n \rceil$ for some $n$.)
If $|z| = N^b$ and $N = ik$, note that since $i \ge 60r$, $k \le
|z|^{1/b}/60r$.
$M_j'$ then computes
$$M_j(1^{2^k},(f_{\lceil \log
  n \rceil}(x_1),\ldots,f_{(i-1)\lceil \log n \rceil}(x_{i-1}), z,
f_{(i+1)\lceil \log   n \rceil}(x_{i+1}),\ldots, f_{\lceil \log n \rceil^2}(x_{\lceil\log n}\rceil)),$$
where $x_1, \ldots, x_{i-1}, x_{i+1}, \ldots, x_{\log n}$ are
randomly chosen inputs of the appropriate length.  That is, we give
$M_j$ as input a tuple that includes $z$ in the appropriate position
padded out by randomly-chosen elements of the right form.  Since $M_j$ runs
in time $n^r$
(remember we measure the running time of $M_j$ in terms of the number of $1$'s, that is, $n=2^k$), $M_j$ computes an output given this input in time
$2^{kr}$.  Finally, $M_j'$ checks to see if the
output $y_i$ of $M_j$ inverts $z$ (with respect to $f_{i \lceil \log n
  \rceil}$).  If so, this is what it outputs; otherwise it outputs
some 0.
It is easy to verify that
if $z = f_{i \lceil \log n \rceil}(x)$ for some $x$ and $n$
(in which case $n = 2^k$), then $M_j'$ inverts $z$ with probability
at least $\frac{1}{2^{i\lceil \log n\rceil/10}}$.  Moreover, inverting $z$ can
be done in time $2^{kr} + p(|z|)$, where $p$ is some polynomial in
$z$.  Since $i \ge  60 r$,  $k = \lceil \log n \rceil$, and $|z| = (i
\lceil \log n \rceil)^b$, it can be done in time at most
$2^{i\lceil \log n \rceil/60} + p'(i \lceil \log n \rceil)$, where
$p'$ is some polynomial.  Clearly
$2^{i\lceil \log n \rceil/60} + p'(i \lceil \log n \rceil) \le
2^{i\lceil \log n \rceil/30}$ for all $n$ sufficiently large, so we
can modify $M_j'$ so that it runs in time at most
$2^{i\lceil \log n \rceil/30}$ for all $n$.  This contradicts the
assumption that $f_k$, $k= 1, 2, 3, \ldots$ is a one-way function.

It is easy to check that
$\sum_{i=60r}^{\lceil \log n\rceil^2}\frac{1}{2^{i\lceil \log
    n\rceil/10}} \le 1/n^{3r}$.
By the union bound and our previous claim,
the probability each player will have strictly more than
$60r$ hits is $o(1)$ for $n$ sufficiently large. On the other hand,
the first player can obtain
$60r+1$ hits in time at most $O(n^{60r+2+o(1)})$; for all $i \leq
60r+1$, it can find $y_i$ such that $f_{i \lceil \log
  n\rceil}(y_i)=f_{i \lceil \log n\rceil}(x_i)$ in time $O((i \cdot
\lceil \log n\rceil)^t 2^{i \cdot \lceil \log n\rceil})=O(n^{i+o(1)})$
by exhaustively examining every string in $\{0,1\}^{i \lceil \log
  n\rceil}$, as we assume there exists a TM $M$ that can compute
$f_{i\lceil \log n\rceil}(x_i)$ in time $|x_i|^t$ for some constant
$t$. Therefore, player can switch from $M_1$ to a TM $M_1^*$ that
runs in time $O(n^{60r+2})$
and strictly improve his utility in $G_n$ by an expected amount of
at least $1-n^{-r}$
for all $n$ sufficiently large,
contradicting the definition of Nash equilibrium.

This already shows that there is no $\epsilon$-NE for $\epsilon\le 1$.
An analogous argument shows that given a purported
$\epsilon$-equilibrium $(M_1, M_2)$ with $1<\epsilon \le B$, for all $B>1$, where $M_1$
and $M_2$ run in time at most $n^r$,
we can find a TM $M^*$ for the first player, that runs in time $O(n^{60r+B+2})$ that gives an expected utility
improvement of at least $B+1-n^{-r}$ for player $1$. It follows that for all $\epsilon$, there is no
$\epsilon$-NE.
\end{proof}

Similar ideas can be applied to show there is a 2-player
\emph{extensive-form} polynomial game that has no polynomial
$\epsilon$-NE, where we no longer need to use a type space.
(See \cite{Kuhn} for the definition of extensive-form polynomial game
and polynomial $\epsilon$-NE in extensive-form polynomial games; we
hope that our discussion suffices to give the reader  an intuitive sense.)
In the game $G_n$, instead of the tuple
$(f_{\lceil \log n\rceil }(x^j_1),\ldots,f_{l
  \lceil \log n\rceil}(x^j_l))$ being player $j$'s type, player $j$
chooses $x^j_1, \ldots, x^j_l$ at random and sends this tuple to
player $3-j$.  Again, player $j$ attempts
to invert as many of $f_{\lceil \log n\rceil
}(x^{1-j}_1),\ldots,f_{l \lceil \log n\rceil}(x^{1-j}_l)$ as it can;
they payoffs are just as in the Bayesian game above.
A proof similar to that of Theorem~\ref{thm:one_way} shows that this
game does not have a polynomial NE.

The one-way function game also shows the affect of restricting
strategies to be polynomial-time computable.  Clearly, withtout this
restriction, the game has a trivial NE: all players correctly invert
every element of their tuple.  On the other hand,
consider a modification of the game where in
$G_n$, a player's type consists of a single element $f_n(x_n)$, with
 $x_n$ a bitstring of length $n$ chosen uniformly at
random.
If both players simultaneously invert or fail to
invert $f_n(x_n)$, then both get zero. Otherwise, the player who
correctly inverts gets 1 and the other player gets
$-1$.  Again, it is easy to see that if we take $\mathcal{M}$ to be the family
of all TMs, the only Nash equilibrium is to find $y_n,z_n$ such that
$f_n(y_n)=f_n(z_n)=f(x_n)$. But if $\mathcal{M}$ consists of only
polynomial-time TMs, then it is a polynomial-time NE for both players
to simply output a random string, as neither player can invert $f$ with
 non-negligible probability, and we ignore negligible additive increase
to the utilities of players.

\section{Equilibrium With Respect to Concrete Time Bounds}\label{sec:concrete}

The previous example may lead one to speculate that lack of Nash equilibrium in computational games hinges on asymptotic issues, namely, our ability to consider larger and larger action and type spaces. This raises the question of
what happens if we restrict our attention to games where players are constrained to execute at most $T$ computational steps, where $T>0$ is a fixed integer. It turns out that if the use of randomness is counted as a computational action, then there may not be Nash equilibria, as the following example shows. We assume from now on that $T>2$.

In our computational game, the family of admissible TMs, which we denote by $\mathcal{M_T}$, is the set of all probabilistic TMs whose running time is upper-bounded by $T$. The operation of printing a character takes one computational step, and so does the movement of the cursor to a different location on the tape. The generation of a random bit (or alternatively querying a bit in a designated tape that contains random bits) requires at least one computational step (we allows arbitrary bias of a bit, as it does not effect the proof).

Consider the following 2-player zero-sum normal-form computational game $\mathcal{F}$ between Alice ($A$) and Bob ($B$). For every $n$, $F_n$ is the same game $F$.  The action space of each player is $\{0,1\}^T$. By our choice of $\mathcal{M_T}$, both players are constrained to perform at most $T$ computational steps.
The game proceeds as follows. $A$ and $B$ use TMs $M_A,M_B \in \mathcal{M_T}$ respectively, to compute their strategies. $M_A$ outputs a single bit $a_1$. $M_B$ outputs $b_1 \in \{0,1\}$. Based on $a_1$ and $b_1$, a game of \emph{matching pennies} is played. Namely, if $a_1=b_1,$ $A$ gets 1, otherwise $B$ gets 1. In the second phase of the game, the TM of each player prints as many characters as possible without violating the constraint of performing at most $T$ steps. If the final number of characters is the same for both players, then both get a payoff of $0$ for the second phase. Otherwise the player with a larger number of printed characters gets an additional bonus of 1.

\begin{theorem}\label{thm:pennies}
The computational game $\mathcal{F}$ does not have an $\epsilon$-$\mathcal{M_T}$-NE, for all $\epsilon<1$.
\end{theorem}
\begin{proof}

Assume, by way of contradiction, that $(M_A,M_B)$ is a Nash equilibrium for $\mathcal{F}$. Since TMs in $\mathcal{M_T}$ are constrained to query at most $T$ bits, it follows that the strategy computed by $M_A$ (or $M_B$) given $1^n$, will be the same for all $n>T$. As the outcomes of the games $F_m$, $m\leq T$, do not effect, by our definition of NE in computational games, whether $(M_A,M_B)$ is an equilibrium, we can assume w.l.o.g that both $M_A$ and $M_B$ compute the same strategy (whether mixed or pure) in all games $F_n.n \geq 1$.

Suppose that one of the players uses randomization. Assume this is player $A$. Namely, $M_A$ generates a random bit before outputting $a_1$.
Then $A$ can guarantee a payoff for the first phase (the matching pennies game) that is no smaller than his current payoff by choosing a TM $M'_A$ that outputs a \emph{deterministic} best response $a_1$ against against the strategy of $B$ in the matching penny game. Observe that we can assume that $a_1$ is ``hardwired" to $M'_A$. In particular outputting $a_1$ can be done in a single computational step. Then $A$ can print strictly more $1$'s in the second phase of the game by configuring $M'_A$ to print $T-1$ $1$'s (which can be done in $T-1$ steps). If $B$ prints $T-1$ in the second phase of the game, we have that $A$ can increase its payoff in $F_n$ for all $n$ by switching to $M'_A$. If, on the other hand, $M_B$ prints less than $T-1$ characters in the second step, an analogous argument shows that $B$ can strictly increase its payoff in $F_n$ for all $n$, by using a TM that runs in at most $T$ steps. In any event, we get a contradiction to the assumption that $(M_A,M_B)$ is a NE.

Suppose now that $A$ does not use randomization. In this case, it follows immediately by the definition of matching-pennies that either $A$ or $B$ can strictly improve their payoff in the first phase of $F_n$ for all $n$, by outputting the (deterministic) best response to their opponent and printing $T-1$ characters afterwards. As before, we can assume this response is hardwired to the appropriate TM, such that outputting it consumes one computational step, allowing players to print $T-1$ characters in the second phase of the game.

Finally, it is not difficult to verify that the argument above establishes that $\mathcal{F}$ does not have an $\epsilon$-NE for $\epsilon$-NE for all $\epsilon \in (0,1)$.
\end{proof}

One might wonder whether the non-existence of NE in computational games follows from the fact that we are dealing with an infinite sequence of games with infinitely many possible TMs (e.g., $|\mathcal{M}|=\infty$). Nash Theorem regarding the existence of NE requires that the action space of every player is finite; without this requirement a NE may fail to exist. Hence it is natural to ask whether limiting $|\mathcal{M}|$ to be finite (for example, taking $\mathcal{M}$ to be the family of all TMs over a fixed alphabet with at most $S$ states for some bound $S$) may force the existence of NE in computational games.
Theorem~\ref{thm:pennies} illustrates that this is not the case:
$\mathcal{F}$ will not have a NE even if we take $\mathcal{M}$ to consist only of TMs whose number of states is upper bounded by a large enough positive number $S$ ($S$ should allow for using the TM that is hardwired to output the appropriate best response in the matching pennies game and print $T-1$ characters in the second phase).  The reason why NE does not exist despite the finiteness of $\mathcal{M}$, is that in contrast to ordinary games, where a mixed actions of best responses is a best response, in our setting this is not necessarily true: mixing over actions may consume computational resources, forcing players to choose actions that are suboptimal when using randomized strategies.

\section{Conclusion}

We have examined sequences of games where TMs compute strategies of players. We demonstrated that NE for polynomial time players may not exist. The implications are twofold: first it hints that in competitive situation between agents that can use more running times or faster algorithms to improve their outcomes, there may evolve a
``computational arms race" which does not admit NE. Second, it entails that classic notions in game theory such as best-response should be treated carefully when considering sequence of games with time-bounded players. Specifically, these notions may fail to exist even if we ignore negligible effects on the utility of players (as in the definition of polynomial time NE).

There are several factors that may imply the existence of equilibrium
when considering a sequence of games. For example, when the number of
states of the TM used by players is bounded there may exist an
equilibrium. Studying properties of games or TMs used by players that result with equilibrium in sequences of games seem as an interesting direction for future research. It might also prove worthwhile to study the effect of limiting resources other than time such as space or the amount of randomness used by players.

\end{document}